\documentclass[journal,10pt]{IEEEtran}
\usepackage[pdftex]{graphicx}
\usepackage{multirow}
\usepackage{amsmath}
\usepackage{mathtools}
\usepackage{epstopdf}
\usepackage{authblk}
\usepackage{cite}
\usepackage{cleveref}
\usepackage{latexsym}
\usepackage[justification=centering]{caption}
\usepackage{amsthm}
\usepackage{balance}
\usepackage{relsize}
\usepackage[nodisplayskipstretch]{setspace}
\usepackage{color}
\usepackage{amssymb}
\usepackage{eucal}
\usepackage{url}
\usepackage{booktabs,subcaption,amsfonts,dcolumn}
\newcommand\norm[1]{\left\lVert#1\right\rVert}
\usepackage{float}
\usepackage{scalefnt}
\usepackage{geometry}
\newtheorem{prop}{Proposition}
\usepackage[nodisplayskipstretch]{setspace}
\usepackage{color}  
\definecolor{Mypink}{RGB}{255,0,255}
\definecolor{Myorange}{RGB}{255,102,0}
\definecolor{Mygreen}{RGB}{0,153,0}
\definecolor{Myblue}{RGB}{0,0,255}
\newcolumntype{d}[1]{D..{#1}}
\usepackage{lipsum}

\DeclareMathAlphabet\mathbfcal{OMS}{cmsy}{b}{n}
\usepackage{breqn}

\begin{document}

\title{Space-Time Block Coded Spatial Modulation for  Indoor Visible Light Communications}
\renewcommand\Authfont{\fontsize{12}{14.4}\selectfont}

\author{Shimaa~Naser,~Lina~Bariah,~\IEEEmembership{Senior Member,~IEEE},~Sami~Muhaidat, \IEEEmembership{Senior Member, IEEE},  Mahmoud Al-Qutayri, \IEEEmembership{Senior Member, IEEE},   Murat Uysal, \IEEEmembership{Fellow, IEEE}, and  Paschalis~C.~Sofotasios,~\IEEEmembership{Senior~Member,~IEEE}
\thanks{ This work was supported in part by Khalifa University
under Grant KU/FSU-8474000122 and Grant KU/RC1-C2PS-T2/8474000137.
\textit{Corresponding Author: Paschalis C. Sofotasios}.}
\thanks{S. Naser and L. Bariah are with the Center for Cyber-Physical Systems, Department of Electrical Engineering and Computer Science, Khalifa University, Abu Dhabi 127788, UAE, (e-mails: 100049402@ku.ac.ae; lina.bariah@ieee.org.)}

\thanks{S. Muhaidat is with the Center for Cyber-Physical Systems, Department of Electrical Engineering and Computer Science, Khalifa University, Abu Dhabi 127788, UAE, and also with the Department of Systems and Computer Engineering, Carleton University, Ottawa, ON K1S 5B6, Canada, (e-mail: muhaidat@ieee.org).}

\thanks{M. Al-Qutayri is with the Systems-on-Chip (SoC) Center, Department of Electrical Engineering and Computer Science, Khalifa University, Abu Dhabi 127788, UAE, (e-mail: mahmoud.alqutayri@ku.ac.ae).}
 
\thanks{M. Uysal is with the Department of Electrical and Electronics Engineering, Ozyegin University, 34794 Istanbul, Turkey
(e-mail: murat.uysal@ozyegin.edu.tr).}

\thanks{P. C. Sofotasios is with the Center for Cyber-Physical Systems, Department of Electrical Engineering and Computer Science, Khalifa University, Abu Dhabi 127788,  UAE, and also with the Department of Electrical Engineering, Tampere University, Tampere 33014, Finland, (e-mail: p.sofotasios@ieee.org.)}
}

\maketitle
\begin{abstract}

Visible light communication (VLC) has been recognized as a  promising technology for handling the continuously  increasing quality of service and connectivity requirements in modern wireless communications, particularly in indoor  scenarios. 
In this context, the present work considers the integration of two distinct modulation schemes, namely spatial modulation (SM) with space time block codes (STBCs), aiming at improving the overall VLC system reliability.  
Based on this and in order to further enhance the achievable transmission  data rate, we   integrate  quasi-orthogonal STBC (QOSTBC) with SM, since relaxing  the orthogonality condition of   OSTBC ultimately provides a higher coding rate. 
Then, we  generalize the developed results to any number of active light-emitting diodes (LEDs) and any $M$-ary pulse amplitude modulation size. 
Furthermore, we derive a tight and tractable upper bound for the corresponding bit error rate (BER) by considering a simple two-step decoding procedure to detect the indices of the transmitting LEDs and then  decode the signal domain symbols.  
Notably, the obtained results demonstrate that  QOSTBC with SM  enhances the achievable BER  compared to SM with repetition coding (RC-SM). 
Finally, we compare STBC-SM with both multiple active SM (MASM) and RC-SM in terms of the achievable BER and  overall data rate,   which further justifies  the usefulness of the proposed scheme. 
\end{abstract}

\vspace*{-10pt}

\section{introduction}
\vspace*{-1pt}
\IEEEPARstart{F}uture sixth generation (6G) and beyond mobile communications are envisioned to witness a plethora of novel data-demanding applications that will be typically characterized by particularly stringent operational and quality of service requirements. 
This calls for the development of disruptive technologies that will be capable of catering  for all anticipated demanding requirements, such as substantially higher data rates, ultra-low end-to-end latency, massive scale connectivity, increased robustness and expanded capacity, all under the umbrella of secure and energy efficient green communications \cite{6G,9079587}. A key approach for meeting these requirements is  through the radical shift from   conventional radio frequency (RF) based communications, to more wideband communications using less occupied or unoccupied bands of the  frequency spectrum. 

Based on the above, recent advancements in solid state optical technology have facilitated the emergence of visible light communications (VLC), which is considered a promising breakthrough technology candidate for complementing RF wireless communications. In addition, relatively low deployment cost and ease of implementation can be realized in VLC systems by means of exploiting light emitting diodes (LEDs), which are primarily intended for illumination, in order to perform wireless transmission through intensity modulation and direct detection (IM/DD). 
Yet, despite the superior features of VLC, the performance of such systems is constrained by several factors, such as limited modulation bandwidth and restricted LED peak optical power, as well as by the imposed   positive and real valued nature of  transmitted signals \cite{Jovicic2013}. Motivated by this, extensive research efforts have been devoted in order to enhance the performance of VLC systems by means of developing efficient optical-based modulation, coding, equalization, VLC cooperative communications, and multiple access (MA) schemes\cite{fath2013,Hussein,MA_VLC}. \textcolor{black}{Additionally, the use of burst-mode receivers is particularly suitable for optical multi-access networks, which  are based on packet transmission, as they are capable of enhancing  the transmission data  \cite{Singh2010}. Finally, due to its  high  spectrum  efficiency,  simplified  and  efficient  implementation, as well as  inter-symbol interference mitigation capabilities, orthogonal frequency division multiplexing (OFDM) has been also investigated in the context of VLC. The work in \cite{singh} has considered the performance evaluation of
different types of quadrature-amplitude modulation (QAM) and phase-shift keying (PSK) modulation with three different channel estimation methods in OFDM system aiming to improve the overall throughput performance.  }

It is also recalled  that multiple-input multiple output (MIMO) transmission techniques have been comprehensively studied in RF systems and they have been extensively shown to exhibit  superior performance in terms of system capacity and reliability, compared to conventional single input single output (SISO) systems. Likewise, two MIMO strategies have been also proposed in the literature for VLC systems, namely, space time block codes (STBCs) and spatial multiplexing (SMP). 
On the one hand, SMP schemes are applied in VLC systems in order to accommodate the stringent demands for high data rate communications, by utilizing multiple LEDs in order to transmit different signals simultaneously. However, the simultaneous transmission from all LEDs results in a high inter-channel interference (ICI) and increased detection complexity  \cite{Elgala,Wolniansky,khan}. On the other hand, STBCs schemes are privileged as they achieve an improved reliability, high data rate, ease of implementation, and reduced decoding complexity features \cite{Tarokh, 9252154}. Existing results in the literature quantified the performance  achieved enhancement by applying STBC   over SISO systems in indoor VLC networks \cite{ntogari,Amano,Hwang}. \textcolor{black}{For instance, the work in \cite{Biagi2013} has jointly investigated  STBCs with pulse position  modulation (PPM). Results demonstrated the ability of the proposed scheme to achieve data rates of the order of hundreds of Mb/s in indoor VLC systems, while maintaining good link reliability. Furthermore, the authors in \cite{biagi2020} proposed a quasi-trace-orthogonal-STBC, where the transmitted codewords follow the quasi-trace orthogonal property to provide an enhanced spectral efficiency. Corresponding simulations were experimentally tested for a $4 \times4$ MIMO L-PPM VLC system. The proposed scheme showed an enhanced spectral efficiency performance of up to four times   that of the trace-orthogonal STBC when 8-PPM was used.    }\\
\indent Meanwhile, index modulation is considered  an efficient modulation scheme for enhancing  the spectral and energy efficiency of VLC systems while maintaining a reduced  hardware complexity  \cite{IndexModulation,Mao,Guo2020}. The key difference in index modulation schemes compared to conventional modulation schemes is that additional data can be conveyed over the indices of the major blocks of the communication system through  on/off keying mechanism. 
 These are, for example, the indices of the LEDs, frequency subcarriers, time slots, or a combination of them. 
Based on this, indexing in the spatial domain has been thoroughly investigated in the field of VLC networks \cite{IndexModulation}. Specifically, different forms of spatial index modulation have been proposed and investigated in the open literature, such as,  space shift keying (SSK), generalized space shift keying (GSSK), and spatial modulation (SM) \cite{7405284}.  SSK is the simplest form of spatial domain index modulation where a single LED out of $N_{t}$ LEDs is activated and the data is conveyed over the index of the activated LED. Yet, although SSK has the advantages of providing an interference free communication and simple receiver design, a large number of transmit LEDs is required to achieve higher modulation orders, and hence, high spectral efficiency. 
So, in the generalized case, i.e. in GSSK, $N_{a}$ LEDs are activated simultaneously to boost the spectral efficiency  at the expense of an incurred spatial interference \cite{Gssk}.
In this context, various reported contributions  investigated the performance capability of GSSK in VLC systems. For instance, \cite{Zhang2018} derived a closed form expression for the average symbol error rate (ASER) for direct-code GSSK system and showed that with a low complexity power allocation mechanism, the ASER and spectral efficiency  can be improved significantly. 
Moreover, the authors in
\cite{Wang2018} investigated the   physical layer security of a  GSSK-VLC system and proposed an optimal LED pattern selection algorithm to enhance  the achievable  secrecy performance. A similar analysis was carried out in  \cite{su2021}, wherein a novel spatial constellation design technique was proposed based on GSSK.

It is recalled that SM was firstly introduced by Mesleh $et$ $al.$ in \cite{Mesleh}, as an efficient scheme to reduce the ICI introduced by SMP and overcome the  spectral efficiency limitations of SSK and GSSK schemes. To achieve this, the principle of SM  is based on conveying information in both the signal and spatial domains.  \textcolor{black}{ A detailed comparison between SMP, repetition coding (RC)  and SM was carried out   in \cite{fath2013}. There, it was demonstrated that SM achieves compromised performance between spectral efficiency, error rate, and complexity. 
Likewise, the authors in \cite{Wang2019}  investigated the secrecy rate performance of SM and proposed  channel adaptive selection and a greedy selection schemes for choosing the active transmitter in order to enhance the secrecy rate performance. 
In addition to this, it is also widely known that the performance of SM   in VLC is highly affected by the incurred channel correlation. To overcome this issue, the authors in \cite{Tran2019} proposed an alteration of the orientation of the photo-detector (PD), while an LED grouping based SM method was reported in \cite{gao2020}. Subsequently, the authors in   \cite{Zheng2020} proposed an adaptive SM scheme that balances the trade-off between spectral efficiency and link robustness. In this scheme, the modulation size was dynamically changed based on the channel cumulative distribution function (CDF) of the user. Furthermore, additional studies on adaptive SM were also considered in \cite{Wang_access_2018} and \cite{Wang2018_photo}. Likewise, the integration of SM and non-orthogonal multiple access has been investigated in multi-user MIMO VLC networks as a means of enhancing the overall  spectral efficiency of the system \cite{Yang2019,Amna2020,wu2020}.   }\\ \indent
Following the same concept as in  SSK, a single LED in SM is activated. However, the activated LED additionally sends a positive real-valued symbol, which ultimately results in an enhanced spectral efficiency compared to SSK modulation. In this context and in order to provide a more flexible system design, the  authors in \cite{Lang2015} proposed the activation of $N_{a}$ LEDs out of the total $N_{t}$, while RC was applied in order to transmit the same $M$-ary  symbol through the active LEDs. However, this scheme is still limited in terms of the achievable spectral efficiency since the same symbol is transmitted from all active LEDs. Furthermore, the error rate performance of RC-SM is practically degraded   at higher spectral efficiency values since   higher $M$-ary modulation orders are required in order to maintain a certain high spectral efficiency under the same system setup.  \\ \indent
Multiple active SM (MASM) is a more generalized version of conventional SM that highly enhances the  spectral efficiency by conveying more information in both spatial and signal domains \cite{Rajesh,Li,Hussein}. It is noted that MASM relies on the activation of $N_{a}$ LEDs out of $N_{t}$, whilst  multiple distinct real non-negative $M$-ary symbols are transmitted from each active LED. It was shown in \cite{datta,Nahhal2021} that MASM achieves high spectral efficiency compared to other spatial index modulation schemes. However, its major drawback is the degraded bit error rate (BER) performance, particularly for a large number of activated LEDs and high modulation orders. 
The mentioned BER performance degradation results in a deterioration in the VLC system's throughput. In order to tackle this issue, a bit-to-symbol mapping based on the corresponding  Euclidean distance (ED) was proposed for MASM in \cite{Tran2019_photon}, where  it was shown  that the proposed MASM with ED-mapping outperforms  conventional mapping for MASM. Additionally, the authors in   \cite{Xiao2019} utilized multi-color LED (QLEDs) as transmitters in a MASM system. In their scheme, the optimal QLED pattern was selected according to
the chromaticity property of the LEDs. However, this   scheme exhibits increased complexity, so   signal design schemes with reduced complexity are highly essential. Moreover, constellation design algorithms for single-mode and dual-mode joint MASM were proposed in \cite{kumar2018} for the case of highly correlated VLC channels. On the contrary, an iterative combinatorial symbol design algorithm was proposed in \cite{curry2018}. In this context, the achievable symbol error rate (SER) performance was enhanced for a random symbol set by iteratively adding and removing symbols. 
Finally, the integration of MASM with orthogonal frequency division multiplexing scheme was   considered in   \cite{kumar2019} and \cite{hussein2019}.
Based on this and according to the above discussion, there is an urgent need for the design of efficient modulation schemes in order to ultimately achieve an adequate balance between spectral efficiency enhancement, BER performance, and involved receiver complexity. 
\subsection{Contribution }

As already mentioned, several  MIMO
transmission techniques have been proposed in the open literature to compensate for the spectral loss introduced by
LEDs and the constraints imposed by
intensity modulation direct detection (IM/DD).  
In addition, the MASM concept was recently proposed and was found capable of increasing  the
transmission data rate by means of conveying information in
both signal and spatial domains. However, the main drawback of this scheme is
the degraded error rate performance at higher spectral efficiency
values.
Based on this, the present contribution proposes  space time block coded-spatial modulation (STBC-SM) for indoor VLC systems. It is worth mentioning  here that the integration of STBC with SM is a particularly promising modulation scheme because it can achieve a significant performance enhancement over MASM   by exploiting the advantages of both SM and STBC, while maintaining high energy efficiency and reduced complexity \cite{Xiang,Biagi}. However, an STBC-SM based scenario in the context of VLC has not been investigated in the literature so far. Therefore, the core aim of our work is to propose a comprehensive framework for the investigation of STBC-SM and for quantifying the advantages of integrating SM and STBC scheme in indoor VLC environments. In more details. the main contributions of this paper are summarized as   follows:
\begin{itemize}
    \item We propose an STBC-SM scheme for indoor VLC systems, which offers a trade-off between  spectral efficiency, reliability, and receiver complexity.
    \item We generalize Alamouti STBC-SM  to an arbitrary number of active LEDs. Moreover, we provide a general design for the $M$-ary pulse amplitude modulation (PAM) symbols and their complements. 
    \item We propose the integration of quasi-orthogonal STBCs (QOSTBCs) with SM to further enhance the system spectral efficiency.  
    \item We purpose a new receiver design based on two-step maximum likelihood (ML) decoding process to reduce the complexity associated with the conventional joint ML receivers. 
    \item To corroborate the performance enhancement achieved by STBC-SM, we consider the proposed two-step ML detector to derive the pairwise error probability (PEP) of the active LEDs indices, which is then utilized to derive a union bound for the probability of  incorrect index detection.

    \item Capitalizing on the  derived union bound, we evaluate the probability of incorrect detection of the transmitted symbols, by deriving a tight and tractable bound for the corresponding BER. 

    \item To compare their   spectral efficiency performance, we provide an analysis for the achievable rate of MASM and STBC-SM schemes.
    
    \item Finally, we present extensive analytic and simulation results in order to validate the derived BER bound expression and  to corroborate the BER superiority of the proposed scheme compared to the MASM counterpart as well as the achieved throughput enhancement at high $M$-ary PAM.
    \end{itemize}

To the best of the authors' knowledge, the offered results have not been previously reported in the open technical literature. 
The remainder of this paper is organized as follows: Section II  presents the considered  system and channel models, whereas section III proposes the generalized STBCs and the QOSTBC for VLC systems. Section IV  presents the proposed ML decoding along with  the required steps for the derivation of the BER bound. Capitalizing on this, the corresponding  analytic and simulation results are discussed in section V. Finally,  the paper is concluded in section VI with some useful comments and suggestions.  \\
\textit{Notations:} Throughout the manuscript, unless mentioned otherwise, boldface uppercase and lowercase represent matrices and vectors, respectively. Also, $(\cdot )^{T}$, $|\cdot|$, and $\norm{\cdot}$ denote the transpose, the absolute value, and the norm operations, respectively whereas  \textbf{\textit{I}} denotes the identity matrix. The notations  ${N \choose k}$ and $\lfloor{x}\rfloor_{2^{p}}$  represent respectively the binomial coefficient and the largest integer less than or equal x, that is integer power of 2. Finally, $\mathcal{N}(0,\sigma^2)$ denotes a real-valued Gaussian distribution with zero mean value and variance $\sigma^{2}$. 

\begin{figure}[t]
\centering
\includegraphics[width=250pt, height = 5cm]{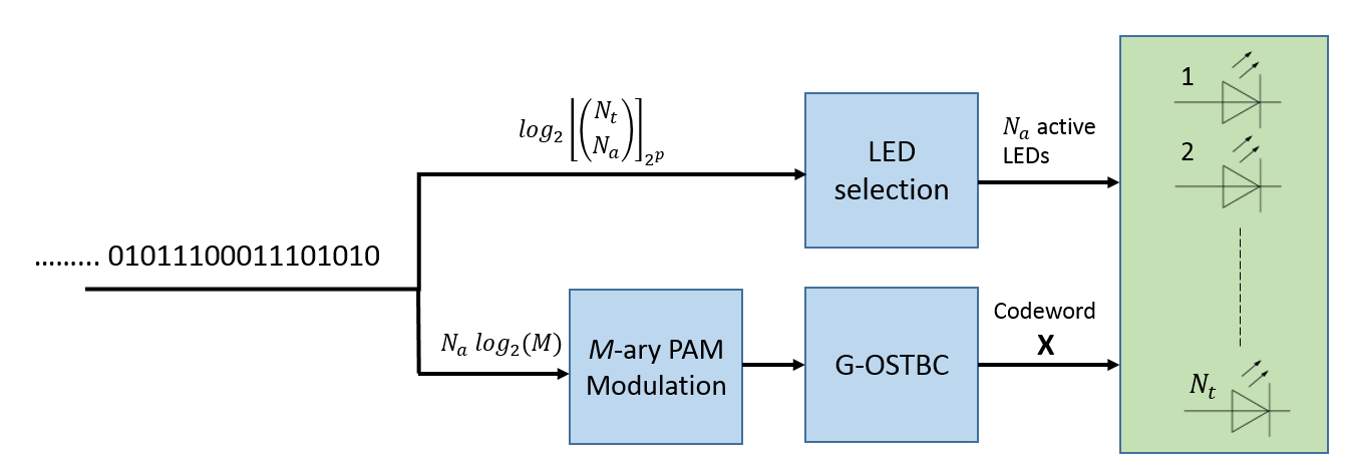}
\caption{ Generalized STBC-SM transmitter}
\label{Fig:STBC-SM}
\end{figure} 

\section{system and channel models}
We consider an indoor VLC multiple-input single-output (MISO) downlink system which consists of $N_{t}$ transmit LEDs, a single PD user ($N_r=1$), and $N_a$ active LEDs. Without loss of generality, the dimensions of the considered room in the present scenario are  $5\times5\times4\;m^3$. The transmitted bit stream is divided into two parts; the first part consists of  $\log _2 (\lfloor{{N_t \choose N_a}}\rfloor_{2^{p}})$ bits, \textcolor{black}{where $\lfloor{{N_t \choose N_a}}\rfloor_{2^{p}}$ denotes the number of combinations of $N_t$ LEDs taken $N_a$ at a time, which is less than  or equal to $N_t \choose N_a$ and  is power of 2.} Those bits are  used to select the LEDs’ indices that will be activated. On the contrary, the other $(R N_a) \log_2(M)$ bits are modulated using $M$-ary PAM modulation and are then encoded using a STBC with coding rate of $R$. The generalized STBC-SM transmitter is depicted in Fig. \ref{Fig:STBC-SM}. Therefore, the received signal can be written as follows
\begin{equation}
\mathbf{y}=  \frac{\eta}{N_a}  \mathbf{X\;h_\ell}+\mathbf{z},
\label{RX}
\end{equation}
where $\eta$ is the PD responsivity and \textcolor{black}{$\mathbf{X}$ $\in \mathbb{R}^{T \times N_a}$ is the transmitted codeword, where  $T$ denotes the number of time slots.} It is worth mentioning that the normalization by $N_a$ is performed in order to guarantee fixed illumination in the room. Moreover,
\textcolor{black}{$\mathbf{z} \in \mathbb{R}^{T \times 1}$} represents the additive white Gaussian noise (AWGN) with zero mean and variance $\sigma^2 = \sigma^2_{shot}+\sigma^2_{th}$, where $ \sigma^2_{shot}$ and $\sigma^2_{th}$ denote the variances of the shot noise and thermal noise, respectively. A detailed noise variance calculation was reported in  \cite{Ghassem2019}. Additionally, \textcolor{black}{ ${\mathbf{h}_\ell} \in \mathbb{R}^{{N_a\times1}}$} denotes the channel gain vector between the $N_a$ active LEDs and the PD,  in which only the line of sight (LoS) component is considered. To this effect, each component of ${\mathbf{h}_\ell}$ can be  expressed as {\cite{Komine2004}}
\begin{equation}
\label{eq:channel}
\small 
 h_{i}= \left\{\begin{matrix}
\frac{A}{d^2_{i}} R_{o}(\varphi_{i}) T_{s}(\phi_{i}) g(\phi_{i}) \cos(\phi_{i}),  & 0 \leqslant \phi_{i} \leqslant \phi_{c} \\
0, & \text{otherwise}
\end{matrix}\right.
\end{equation}
where $A$ denotes the PD area, $d_{i}$ is the distance between the $i$th active LED and the PD, $\varphi_{i}$ is the angle of transmission from the $i$th active LED to the PD, $\phi_{i}$ is the incident angle with respect to the receiver, and $\phi_{c}$ is the field of view (FoV) of the PD. Moreover, $T_{s}(\phi_{i})$ and  $g(\phi_{i} )$ denote the gains of the optical filter and concentrator, respectively. The  $g(\phi_{i} )$ term can be expressed as follows
\begin{equation}
\label{eq:g}
g(\phi_{i})= \left\{\begin{matrix}
\frac{n^2}{\sin^2({\phi_{c}})},  & 0 \leqslant \phi_{i} \leqslant \phi_{c}\\
0, &  \phi_{i}>\phi_{c}
\end{matrix}\right.
\end{equation}
where $n$ represents the refractive index, and $R_{o}(\varphi_{i})$ denotes the Lambertian radiant intensity. \textcolor{black}{ Further details about $R_{o}(\varphi_{i})$ can be found in \cite{Ghassem2019}. }
\vspace{-5pt}
\section {generalized stbc for vlc Systems} 
It is recalled that STBCs have attracted considerable  attention because of their demonstrated  implementation and decoding simplicity, which renders them particularly  capable of exploiting the distinct potentials of MIMO systems. In this context, orthogonal STBCs (OSTBCs) constitute special cases of STBCs which allow single-symbol simplified linear decoding. However, it is noted that OSTBCs in VLC systems can achieve a maximum coding rate of one. Therefore, high coding rate STBCs are required in order to provide higher transmission data rates.  
Thus, in this section we first generalize Alamouti STBC to an arbitrary $M$-ary PAM size.  Subsequently, we consider the general case of an arbitrary number of active LEDs in order to develop the codeword for any optical OSTBCs. Finally, we propose a rate-2 QOSTBC that will enable an enhancement on the achievable transmission data rate.

To that end, starting with the well-known Alamouti STBC, the transmitted codeword from two active LEDs can be expressed as
\begin{gather}
\mathbf{X}_2 = \bigg[
\begin{array} {c c}
 x_1 & x_2 \\
\bar{x_2} & x_1
\end{array}
\bigg ],
\end{gather}
where $x_1$ and $x_2$ are real and positive PAM symbols. Moreover, $\bar{x_i}$ denotes the complement of $x_i$, which  for the special case of on-off keying (OOK) is calculated as $\bar{x_i}=I_p-x_i$, where $I_p$ is the mean optical power. Also,  $x_i$ and $\bar{x}_i$ can be generalized into any PAM size $M$, so the intensity levels for the case of $M$-ary PAM  are defined as follows 
\begin{equation}
\label{eq:PAM}
x_i \in \frac{2 i I_p}{M-1}  , \; i=0,1,...,(M-1).
\end{equation}

It is noted here  that when the LEDs transmit bits corresponding to the first level (i.e. $x_i=0$), they will be incorrectly decoded as inactive. Therefore, in order to overcome this issue and allow the receiver to distinguish the  active LEDs, the intensity levels are modified as follows
\begin{equation}
\label{xi}
x_i \in \frac{2 i I_p}{M+1}  , \; i=1,2,...,M
\end{equation}
and the complement of the symbol $x_i$ is 

\begin{eqnarray}
\label{xi_bar}
\bar{x_i}&=&-x_i + \frac{2 I_p}{M+1} +\frac{2MI_p}{M+1}  \\\nonumber \\\nonumber
&=& -x_i + 2I_p ,  \;   \;  i=1,2,...,M.
\end{eqnarray}
\subsection{Optical Orthogonal STBCs}
Alamouti STBC can be generalized to produce an $N_a\times N_a$ OSTBC of rate $R=1$, which can be expressed as 
\begin{gather}
\mathbf{X_{N_a}} = \bigg[
  \begin{array} {c c}
     \mathbf{X}_{\frac{N_a}{2}}^{1}& \mathbf{X}^2_{\frac{N_a}{2}}\\ 
     (\mathbf{\bar{X}}_{\frac{N_a}{2}}^2)^T &( \mathbf{X}_{\frac{N_a}{2}}^1)^T \mathbf{X_{N_a}} 
   \end{array} \bigg ]  \; \textcolor{black}{\in \mathbb{R}^{N_a\times N_a}},
\end{gather}
where  {$ \mathbf{X}_{\frac{N_a}{2}}^{1}$ }and {$\mathbf{X}_{\frac{N_a}{2}}^{2}$ }are the {(${N_a}/{2}\times {N_a}/{2}$)} OSTBC matrices for the first ${N_a/2}$ and the last ${N_a/2}$ symbols, respectively. The process starts by dividing the $N_a\times N_a$ matrix into blocks of size $N_a/2\times N_a/2$. Then, each block is  also divided into new sub-blocks until the smallest sub-block, i.e the Alamouti codeword is reached.
In the present analysis, and without loss of generality, we consider  the special case of an  $8\times8$ OSTBC matrix.
To this end, we fist start by dividing the  $8 \times 8$ OSTBC matrix into blocks of size  $4 \times 4$ as follows:
\begin{gather}
\mathbf{X_8} = \bigg[
  \begin{array} {c c}
     \mathbf{X}_4^{1}& \mathbf{X}^2_4\\ 
     ({\mathbf{\bar{X}}_4^2})^T & (\mathbf{X}_4^1)^T
   \end{array} \bigg ],
   \label{8STBC}
\end{gather}
where $\mathbf{X}_4^{1}$ and $\mathbf{X}^{2}_{4}$ are  $4 \times 4$ codewords for ($x_1, x_2, x_3, x_4$) and ($x_5,x_6,x_7,x_8$),  respectively. Then, each block is divided into a smaller $2 \times 2$ OSTBC sub-blocks, which yields \par

\begin{gather}
\mathbf{X^1_4} = \Bigg[
  \begin{array}{c c} 
  \mathbf{X}_2^{1}& \mathbf{X}^{2}_{2}\\ 
({\mathbf{\bar{X}}_2^{2}})^T & (\mathbf{X}_2^{1})^T
   \end{array}  \Bigg]
\;
, \mathbf{X^2_4} =\Bigg[
  \begin{array}{c c}
\mathbf{X}_2^{3}& \mathbf{X}^{4}_{2}\\ 
({\mathbf{\bar{X}}_2^{4}})^T & (\mathbf{X}_2^{3})^T
\end{array} \Bigg],
\label{4STBC}
\end{gather}
where  $\mathbf{X}_2^{1}$ and $\mathbf{X}^{2}_{2}$ are Alamouti codewords for ($x_1, x_2$) and ($x_3, x_4$), respectively. On the contrary, $\mathbf{X}_2^{3}$ and $\mathbf{X}^{4}_{2}$ are Alamouti codewords for ($x_5, x_6$) and ($x_7, x_8$), respectively.  
Meanwhile, OSTBC codes for real symbols of rate $R=1$ for $N_a=$3, 5, 6, and 7 can be obtained by eliminating some of the columns of the square OSTBC matrices defined earlier for $N_a =4$ and 8 in (\ref{4STBC}) and (\ref{8STBC}), respectively. For example, a $4\times3$ OSTBC can be generated by eliminating the last column of the $4\times4$ OSTBC, yielding 
\begin{gather}
\mathbf{X_3}=
  \begin{bmatrix}
   x_1 & x_2 & x_3 	\\
\bar{x_2} & x_1& \bar{x_4}	\\
\bar{x_3} & x_4 & x_1 	\\
\bar{x_4} & \bar{x_3} & x_2 
   \end{bmatrix} .
\end{gather}
Importantly, the generated codewords satisfy the orthogonality condition $\mathbf{X_i^TX_i}=\mathbf{I}$. Therefore, they can be exploited for   integrating STBC with SM in order to simplify receiver design and to allow the transmission of  higher order modulation, as presented in detail in the following section. 
\vspace{-5pt}
\subsection{Optical Quasi-Orthogonal STBCs}

Despite their simplified ML detection, OSTBCs in VLC can achieve maximum coding rate of 1. Therefore, in order to provide higher rates the orthogonality is relaxed in order to transmit more symbols in each time slot through developing QOSTBCs.  It is worth mentioning that QOSTBCs allow a trade-off between higher rates and decoding complexity; this is because parallel ML detectors are needed to jointly decode pairs (or more) of the transmitted symbols  instead of single-symbol decoding as in OSTBCs. In what follows, we provide an example of a QOSTBC that achieves a rate $R=2$, which is then utilized to be integrated with SM. 

\subsubsection{ Rate-2 QOSTBC}

  A rate-2 QOSTBC  for real symbols and two transmit LEDs is expressed as     
\begin{gather}
\mathbf{X}=
  \begin{bmatrix}
{  a x_1+b x_3} &{ a x_2+b x_4}\\
{c \bar{x}_2+d \bar{x}_4} & {c x_1+d x_3 }	\\
   \end{bmatrix} .
   \label{QOSTBC_matrix}
\end{gather}

\begin{figure*}
\centering
\begin{equation}
\mathbf{H}^T_{eff}\mathbf{H}_{eff}=
\begin{bmatrix}
{ (a h_1)^2 +(c h_2)^2}& {0} &{abh^2_1 +cdh^2_2}& {0} \\ \\ {0} & {(ch_1)^2 +(ah_2)^2}& {0} &  {cdh^2_1+abh^2_2} \\ \\
   {abh^2_1 +cdh^2_2 }& {0} & {(bh_1)^2+(dh_2)^2} & {0}\\ \\
{ 0 }& cdh^2_1+abh^2_2  & 0 & (dh_1)^2+(bh_2)^2 \\
    \end{bmatrix}  
    \label{qostbc}.
    \end{equation}
    \hrulefill
\end{figure*}
Of note, equation  (\ref{QOSTBC_matrix}) is   simply a linear combination of two optical Alamouti STBCs. Moreover, the parameters $a, b, c, $ and $d$ need to be carefully chosen in order to maximize the coding gain, while they also satisfy a fixed average transmitted optical power, i.e.,
$a+b= c+d =1$. 
In order to simplify the detection process, we can choose $a=c$ and $b=d$. Therefore, the multiplication of the effective channel gain matrix and its transpose can be expressed as in (\ref{qostbc}),  at the top of the next page. It is also worth noting that the values of the design parameters $a, b,c,$ and $d$ need to be optimized in order to obtain best detection performance.  
Also, it can be noticed from (\ref{qostbc}) that the joint ML decoding of the symbols $x_1$ and $x_3$ will be independent of the other two symbols, $x_2$ and $x_4$.
\section{Proposed ML Decoding for STBC-SM Scheme }

In this section, we investigate the integration of STBC with SM in order to enhance the achievable  VLC system performance in terms of reliability, throughput and complexity. 
To this end and by assuming that ML detection is utilized at the receiver side,  the  detector performs joint detection to deduce the received signals over the space and signal domains. In particular, the receiver will jointly detect the LED index and the transmitted symbol according to the following criterion 
\begin{equation}
[\hat{\ell},\hat{\mathbf{X}}]=\arg\min_{\tilde{\ell}\in\mathbb{\ell},\tilde{X}\in\mathbb{X}} \left\Vert\mathbf{y}- \frac{\eta}{N_a}\tilde{\mathbf{X}}{\mathbf{h_{\tilde\ell}}}\right\Vert^2,
\end{equation}
where $\ell$ denotes the index of the transmitting LEDs and $\mathbf{X}$ is the transmitted codeword. It also recalled here that  joint ML detection requires the search over all ${M^{(R\times N_a)} }\times \lfloor{{N_t \choose N_a}}\rfloor_{2^{p}}$ combinations, which results in a significantly increased receiver complexity. 
Yet, the orthogonal property of the generated OSTBC codewords and being of rate $R=1$ allows the symbols $x_i$'s to be decoded independently. This ultimately reduces the complexity to $M\times  N_a  \times \lfloor{{N_t \choose N_a}}\rfloor_{2^{p}}$. 
To this effect and  aiming at reducing receiver complexity,  we propose a detection mechanism in which the detection process is broken down into two main stages. The first stage comprises a conditional ML detection, in which conditioned on $\mathbf{X}$, the receiver detects only the indices of the transmitting LEDs as   follows:  
\begin{equation}
\hat{\ell}=\arg\min_{\tilde{\ell}\in\mathbb{\ell} }\left\Vert \min_{X\in\mathbb{X}} \big[ \big(\mathbf{y}- \frac{\eta}{N_a}{\mathbf{X}}{\mathbf{h_{\tilde\ell}}} \big) \big|\mathbf{X} \big ]\right\Vert^2. 
\label{index}
\end{equation}
In the second stage, the orthogonality of the generated STBC codewords is utilized to decouple the transmitted symbols. Based on this and by assuming perfect detection of the LEDs indices, and therefore perfect estimation of the channel matrix, the transmitted codeword can be evaluated as follows: 
\begin{equation}
\label{eq:signal}
\mathbf{H^T_{\ell_{eff}}}\mathbf{y}= \mathbf{\check{y}}= \norm{\mathbf{h_{\ell}}}^2 \; \frac{\eta}{N_a} \;\mathbf{{I}} \; \mathbf{x} + \mathbf{z},
\end{equation}
where, $\mathbf{x}$ is the transmitted symbols vector, $\mathbf{x}=[x_1, x_2,...x_{\rm N_a}]^T$ and \textcolor{black}{$\mathbf{z} \in \mathbb{R}^{T\times 1}$} is the AWGN vector, in which each component has a zero mean and a variance of $\sigma^2. \norm{\mathbf{h_{\ell}}}^2$. Also,  \textcolor{black}{$\mathbf{H_{\ell_{eff}}} \in \mathbb{R}^{T \times N_{ov}}$ is the effective channel gain matrix for the OSTBC, where $N_{ov}$ denotes the overall number of symbols transmitted in $T$ time slots.} Finally, in order to detect the transmitted symbols,  a conventional signal domain ML detector is utilized to decide on each transmitted symbol separately, namely
\begin{equation}
\hat{x_i}=\arg\min_{\tilde{x_i}\in{\chi}} \bigg| {\check{y_i}- \norm{\mathbf{h_{\ell}}}^2 \; \frac{\eta}{N_a} \tilde{ x_i} } \bigg |^2.
\label{symbol}
\end{equation}
It is also emphasized here that the proposed detector reduces the receiver complexity to $M \times Na+ \lfloor{{N_t \choose N_a}}\rfloor_{2^{p}}$. Furthermore, if the  QOSTBC defined in (\ref{QOSTBC_matrix}) is integrated with SM, the complexity is reduced to $2\times(M \times Na)+ \lfloor{{N_t \choose N_a}}\rfloor_{2^{p}}$.

In what follows, we evaluate the BER performance of the proposed receiver for STBC-SM-based indoor VLC system.
\vspace{-15pt}
\subsection{BER Analysis}
The probability of incorrect detection of the transmitted symbols can be evaluated as
\begin{equation}
P_e =  P_{\mathbf{X}}(\mathbf{X}| \ell \neq \hat{\ell} ) P_{\ell} + P_{\mathbf{X}}(\mathbf{X}|\ell =\hat{\ell}) (1-P_{\ell}),
\label{total_error}
\end{equation}
where $P_{\ell}$ is the probability of incorrect index detection and  $P_{\mathbf{X}}(\mathbf{X}|\ell=\hat{\ell})$ and $P_{\mathbf{X}}(\mathbf{X}| \ell \neq \hat{\ell})$ denote the probability of incorrect codeword detection conditioned on correct and incorrect index detection, respectively. Due to mathematical intractability of (\ref{total_error}) and given that $ P_{\mathbf{X}}(\mathbf{X}| \ell \neq \hat{\ell} )$ is rather large, i.e. close to unity, the total error probability can be simplified as
\begin{equation}
P_e \leqslant  P_{\ell} + P_{\mathbf{X}}(\mathbf{X}|\ell =\hat{\ell})   (1-P_{\ell}).
\label{upper_BER}
\end{equation}

Therefore, obtaining the BER from (\ref{upper_BER}) requires the derivation of two expressions for  $P_{\ell}$ and $ P_{\mathbf{X}}(\mathbf{X}|\ell =\hat{\ell})$. To this effect and given that $P_{\ell}$ is conditioned on $\mathbf{X}$, we resort to the corresponding PEP, which represents an accurate metric to quantify the BER performance  of a wireless system \cite{Lina2019}. 
In fact,  the derivation of the PEP constitutes the fundamental step for the derivation of a union bound on the BER. 

\newtheorem{pro}{Proposition}[section]
\begin{prop}

The following closed-form expression is valid for the  union bound of the BER for $P_{\ell}$: 
\begin{equation}
\centering
P_\ell =  \frac{1}{N 2^N}\sum\limits_{\forall \ell} \sum\limits_{\ell \neq \hat{\ell}} D(\ell ,\hat{\ell}) {Q} \bigg(\sqrt{ \frac{\eta^2}{4N_a^2 \sigma^2} \norm{\mathbf{{X}} \mathbf{h}_{{\ell}} - \mathbf{\hat{X}}_{min} \mathbf{h}_{\hat{\ell}} }^2} \bigg)
\label{union_bound_final}
\end{equation}
where $D(\ell ,\hat{\ell})$ is the number of different  bits between two distinct symbols.
\end{prop}

\begin{proof}
With the aid of  PEP, the upper bound to the probability of a union of all events of decoding  $\ell$ as $\hat{\ell}$ can be represented as 
\begin{equation}
\centering
P_\ell =  \frac{1}{N 2^N}\sum\limits_{\forall \ell} \sum\limits_{\ell \neq \hat{\ell}} D(\ell ,\hat{\ell}) P(\ell \longrightarrow \hat{\ell}), 
\label{union_bound_definition}
\end{equation}
where $P(\ell \longrightarrow \hat{\ell})$ is the PEP for index detection, which can be written as   

\begin{equation}
\centering
P(\ell \longrightarrow \hat{\ell}) = P\bigg( \norm{\mathbf{y}-\frac{\eta}{N_a} \mathbf{\hat{X}}_{min} \mathbf{h}_{\hat{\ell}} }^2 \leqslant \norm{\mathbf{y}-\frac{\eta}{N_a} \mathbf{{X}}_{min} \mathbf{h}_{{\ell}} }^2 \bigg). 
\label{PEP_not_simplified1}
\end{equation}
After substituting the received signal $\mathbf{y}$ from (\ref{index}) in (\ref{PEP_not_simplified1}), the PEP can be written as 
\begin{equation}
\centering
P(\ell \longrightarrow \hat{\ell}) = P\bigg( \norm{\frac{\eta}{N_a} \mathbf{{X}}\mathbf{h}_{\ell}+ \mathbf{n}-\frac{\eta}{N_a} \mathbf{\hat{X}}_{min} \mathbf{h}_{\hat{\ell}} }^2 \leqslant \norm{\mathbf{n} }^2 \bigg). 
\label{PEP_not_simplified2}
\end{equation}
Based on this and after some   algebraic manipulations,  the PEP can be expressed in closed form  in terms of the one dimensional Gaussian $Q-$function, namely
\begin{equation}
\small
P(\ell \longrightarrow \hat{\ell}) = {Q} \bigg(\sqrt{ \frac{\eta^2}{4N_a^2 \sigma^2} \norm{\mathbf{{X}} \mathbf{h}_{{\ell}} - \mathbf{\hat{X}}_{min} \mathbf{h}_{\hat{\ell}} }^2} \bigg). 
\label{PEP_simplified}
\end{equation}
Finally, the obtained PEP in (\ref{PEP_simplified}) can be substituted in (\ref{union_bound_definition}), which yields the closed form representation  in (\ref{union_bound_final}), which completes the proof. 
\end{proof}
By leveraging the orthogonality feature of the STBC codewords defined in III.A, transmitted symbols can be decoded separately using the ML detector in (\ref{symbol}). Therefore, the BER expression of M-ary PAM can be utilized in evaluating the probability of incorrectly detecting $\mathbf{X}$, namely \cite{Ghassem2019}
\begin{equation}
 P_{\mathbf{X}}(\mathbf{X}|\ell =\hat{\ell}) = \frac{2 (M-1)}{M \log_2(M)}Q\bigg(\sqrt{\frac{\eta^2 \norm{\mathbf{ h}_l}^2 I^2_{p}}{N^2_a(M+1)^2 \sigma^2}}\bigg).
 \label{BER_codeword}
 \end{equation}

\begin{figure*}[ht]
\centering
\begin{dmath}
\label{final_expression} 
 P_e \leqslant   \frac{1}{N 2^N}\sum\limits_{\forall \ell} \sum\limits_{\ell \neq \hat{\ell}} D(\ell ,\hat{\ell}) {Q} \bigg(\sqrt{ \frac{\eta^2}{4N_a^2 \sigma^2} \norm{\mathbf{{X}} \mathbf{h}_{{\ell}} - \mathbf{\hat{X}}_{min}  \mathbf{h}_{\hat{\ell}} }^2} \bigg) + \frac{2 (M-1)}{M \log_2(M)}Q\bigg(\sqrt{\frac{\eta^2 \norm{ \mathbf{h}_l}^2 I^2_{p}}{N^2_a(M+1)^2 \sigma^2}}\bigg) \times 
 \left[1- \frac{1}{N 2^N}\sum\limits_{\forall \ell} \sum\limits_{\ell \neq \hat{\ell}} D(\ell ,\hat{\ell}) {Q} \bigg(\sqrt{ \frac{\eta^2}{4N_a^2 \sigma^2} \norm{\mathbf{{X}}  \mathbf{h}_{{\ell}} - \mathbf{\hat{X}}_{min}  \mathbf{h}_{\hat{\ell}} }^2} \bigg)\right]. 
\end{dmath}
\hrulefill
\end{figure*}
Therefore, with the aid of (\ref{upper_BER}), (\ref{union_bound_final}) and (\ref{BER_codeword}), the overall  BER upper bound is given by the explicit expression in  (\ref{final_expression}), at the top of the next page. 

\subsection{Achievable Rate Analysis for MASM and STBC-SM}

Motivated by the advantages of MASM in improving the system's spectral efficiency, in this section we derive the  rate expression for MASM and STBC-SM schemes and compare them in terms of their achievable spectral efficiency.
To this end, it is recalled that the spectral efficiency for MASM depends on the activation patterns and the modulation size, namely  
\begin{equation}
\eta_{masm} = N_{a}  {\log}_{2}(M)+ \log_{2}\bigg\lfloor{{N_t \choose N_a}}\bigg\rfloor_{2^{p}}.
\end{equation}
Based on this and given that the rate of correctly detected bits experiences severe degradation  for the case of a large number of active LEDs and modulation size, the achievable rate of MASM can be evaluated in terms of the error probability  as  
\begin{equation}
R_{masm} = (1-P^{masm}_e) \times \eta_{masm}, 
\label{rate_MA_SM}
\end{equation}
where $P^{masm}_e$ is the BER for a MASM system, which is expected to be particularly high for large $M$ values. Following the same steps,  the achievable rate of STBC-SM is given by 
\begin{equation}
R_{stbc-sm} = (1-P_e) \times \eta_{stbc-sm}, 
\label{rate_STBC_sm}
\end{equation}
where $\eta_{stbc-sm}$ is the spectral efficiency of STBC-SM,  which can be expressed as
\begin{equation}
\eta_{stbc-sm} = \frac{{R\times N_a \log}_{2}(M)+\log_2(\lfloor{{N_t \choose N_a}}\rfloor_{2^{p}})}{T} 
\label{SE_STBC}
\end{equation}
 Therefore, it is noted that since  STBC-SM scheme achieves a better BER performance for high $M$ values compared to MASM and RC-SM,  STBC-SM  outperforms both schemes in terms of the corresponding BER, and hence the overall achievable rate. In addition, since the proposed QOSTBC   in (\ref{QOSTBC_matrix}) has a rate of 2, when integrated with SM it will require lower modulation size compared to RC-SM to achieve the same spectral efficiency. This will be reflected on the error rate performance at high spectral efficiency values, as  it is demonstrated in the next section.

It is noted here that evaluating the error rate performance with respect to the average electrical received signal-to-noise ratio (SNR) would disregard the individual path loss of the different setups and activated LEDs pattern. Therefore, in order to guarantee fair comparison, we opt to evaluate the error rate performance of the proposed scheme with respect to the transmit SNR, which is defined as the ratio of the average symbol energy against the noise power
spectral density \cite{fath2013,fath_20132}. 
It is also noted that the value of the transmit SNR in VLC systems is considerably higher than the one encountered in typical RF communication systems as it often exceeds 100 dB \cite{He2015}. This is attributed to the small value associated with the noise power spectral density $N_0$ \cite{Yin2016}. Specifically, by neglecting the photodetector dark current, $N_0$ is expressed as $N_0 \simeq q\;I_B $, where $q=1.6e{-19}$ is the charge of electron and $I_B$ is the background noise current that takes typical values in the order of $\mu A$ \cite{Stavridis_2015,Ghassem2019}.

\section {Numerical Results} \label{sec:models}

In this section, we capitalize on the offered results to analyze the achievable performance of the proposed STBC-SM scheme. Hence, the validity of the derived analytic expressions in (\ref{final_expression}), (\ref{rate_MA_SM}) and (\ref{rate_STBC_sm}) is first justified through extensive comparisons with respective results from computer simulations. To that end and without loss of generality, we consider an indoor VLC environment with dimensions  of $5\times5\times4\;m^3$, in which the locations of both the involved LEDs and the user are illustrated in Fig. \ref{Fig:room}. 
In addition, it is assumed that the transmit LEDs radiate downward from the ceiling to the floor. On the contrary, the receiver position in the room is considered to be fixed throughout simulations and is assumed to be located   at  height of 0.8 m from the floor. Based on this, the corresponding channel gain can be also  considered   fixed.  The receiver is also assumed to be oriented upward to the ceiling. Moreover, the LEDs semi-angle at half power $\varphi_{1/2}$ and the FoV of the PD are set to $60^o$, which is a typical value for commercially-available high-brightness LEDs. For convenience, all the parameters  involved  in our simulations are summarized in Table \ref{TableI}.
\begin{figure}[ht]
\centering
\includegraphics[width=250pt, height = 6cm]{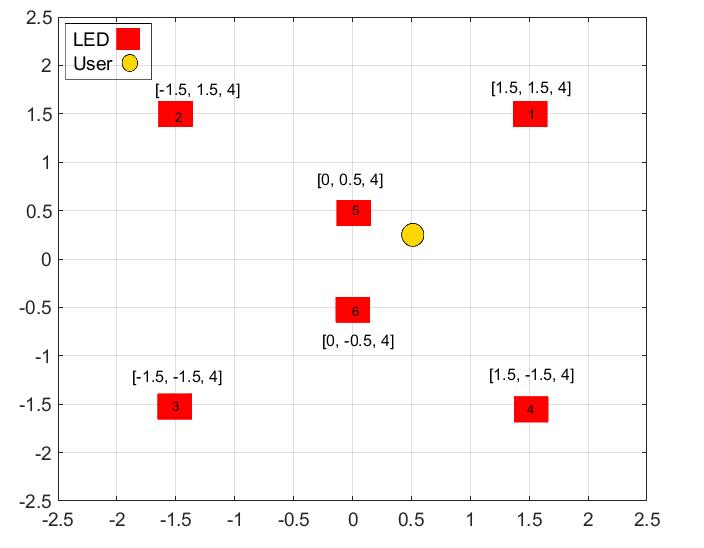}
\caption{ LEDs and user placements in $5\times5\times4\;m^3$ room. }
\label{Fig:room}
\end{figure}

\begin{table}[ht]
\scriptsize
\centering
 \caption{Simulation parameters.}
\begin{tabular} [c]{||p{3cm}|p{1.5cm}|p{2cm}||}
\hline 
\textbf{Parameter}&\textbf{Symbol}&\textbf{Value} \\ \hline  \hline
Room dimensions&- &$5\times5\times4 m^3$ \\ \hline
LED beam angle& $\varphi_{1/2}$ &	$60^o$\\ \hline
PD area& $A$ &	$1\; cm^2$\\ \hline
Refractive index of PD& $n$ & $1.5$\\ \hline
Gain of optical filter&	$T_s(\phi_{i})$  &1\\ \hline
FoV of PD & $\phi_c$	& $60^o$\\ \hline
PD responsitivity& $\eta$ &	$1 A/W$\\ \hline
\end{tabular}
\label{TableI}
\end{table}

\begin{figure}[ht]
\centering
\includegraphics[width=265pt, height = 7cm]{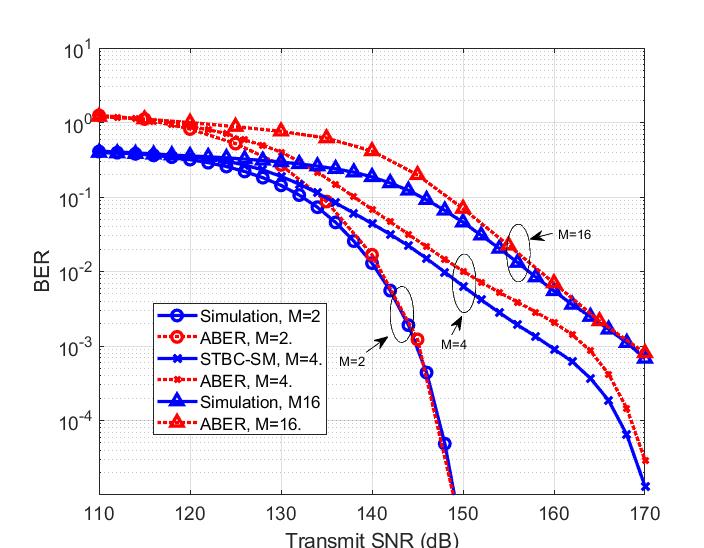}
\caption{ Comparison between the derived analytical results and simulations. $N_T = 6$, $N_a = 2$, and $N_r=1$. }
\label{analytical}
\end{figure}
Fig. \ref{analytical} shows the achievable analytic and simulated  BER  versus the transmit  SNR for different modulation size with $N_T = 6, N_a=2$ and $N_r = 1$. It is noticed that the derived formula in (\ref{final_expression}) forms an upper bound for the BER, which is tight at high SNR values for various modulation size. Additionally, STBC-SM enjoys a good performance in the range of interest of the transmit SNR for VLC system.

Next, a comparison between MASM and STBC-SM systems is considered in terms of the BER performance. The results are demonstrated  in Fig. \ref{Na_2}, where it is shown that  integrating STBC with SM provides a significant BER improvement for large modulation size $M$. On the contrary, the illustrated performance of the corresponding  MASM counterpart exhibits a deteriorated performance at high $M$ values. For instance, for the case of $M=2$, more than 20 dB is required to achieve a BER of $10^{-3}$ using MASM compared to STBC-SM.
Additionally, STBC-SM user exhibits a good BER performance in the typical transmit SNR range for VLC system,  where the received SNR in this case has an offset of 120 dB,  since the channel gain is in the order of $10^{-6}$.
 \begin{figure}[ht]
\centering
\includegraphics[width=265pt, height = 7cm]{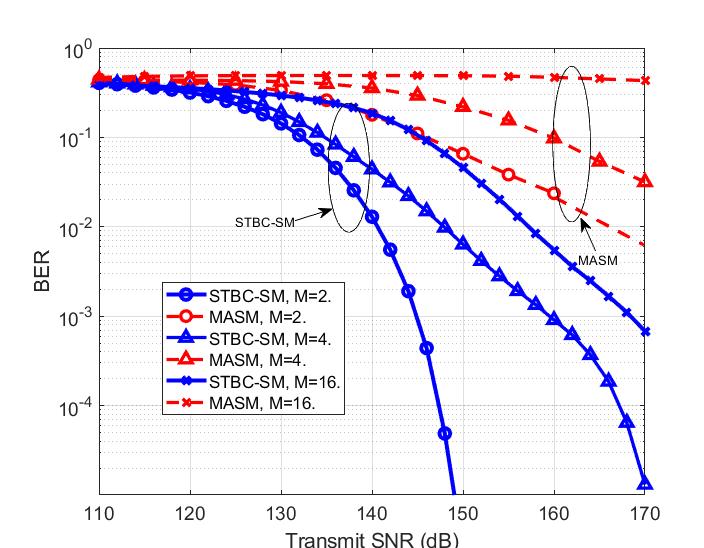}
\caption{ BER performance comparison versus transmit SNR. $N_t = 6$, $N_a = 2$, and $N_r=1$.}
\label{Na_2}
\end{figure}

For fair comparison, we also consider the BER performance for both schemes for a fixed spectral efficiency $\eta$ of 4 bpcu and 5 bpcu,  as depicted in Fig. \ref{SE}. We observe that, STBC-SM provides a reliability enhancement of the considered VLC systems compared to MASM, in the typical SNR range. For instance, for $\eta$ = 4 bpcu at transmit SNR = 160 dB, the BER performance for STBC-SM and MASM is $10^{-4}$ and $10^{-3}$, respectively i.e. a difference of around one order of magnitude.  Moreover, a comparison in terms of the achievable throughput, as defined in (\ref{rate_MA_SM}) and (\ref{rate_STBC_sm}),  for both schemes is illustrated in Fig. \ref{Fig:rate1}. It is noted that in the typical SNR range, the throughput gap between the two schemes decreases as the modulation size increases, which is due to the high BER that occurs in the case of MASM. Meanwhile, using STBC with SM  improves the BER performance,  which is reflected on the achievable data rate as seen from (\ref{rate_STBC_sm}). 

\begin{figure}[ht]
\centering
\includegraphics[width=265pt, height = 7cm]{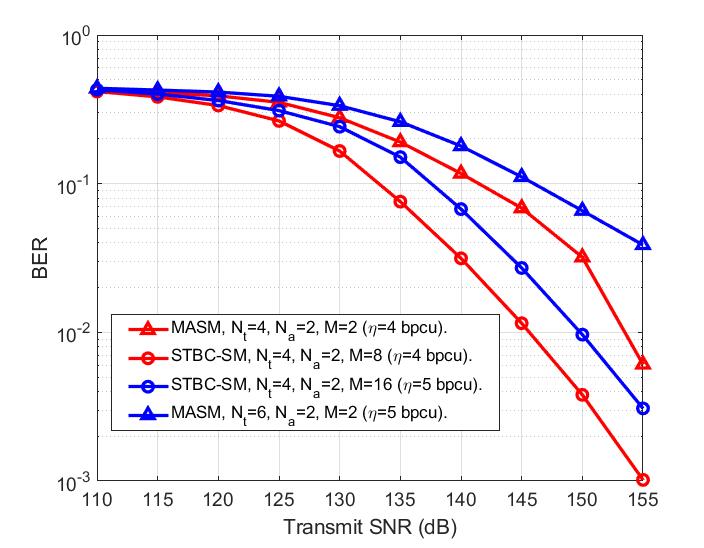}
\caption{BER performance comparison between STBC-SM and MASM for spectral efficiencies of 4 bpcu and 5 bpcu. Red lines for 4 bpcu, blue lines for 5 bpcu.}
\label{SE}
\end{figure}
\begin{figure}[ht]
\centering
\includegraphics[width=265pt, height = 7cm]{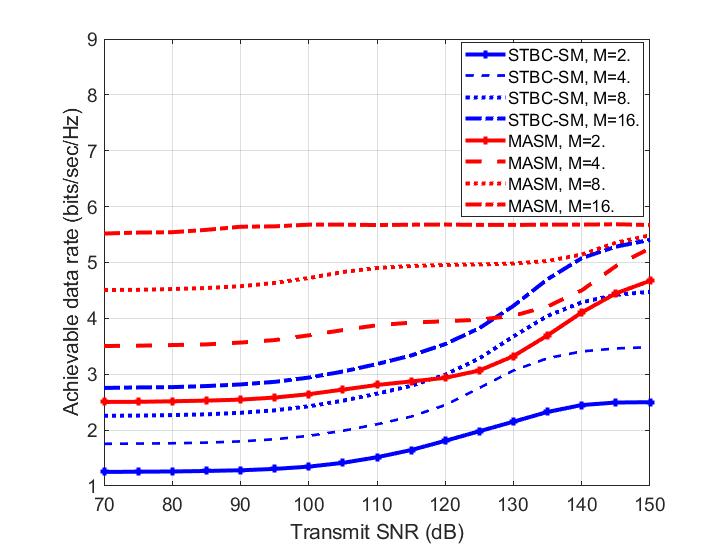}
\caption{The achievable rate performance comparison between STBC-SM and MASM versus transmit SNR. $ N_t = 6$, $N_a = 2$, and $N_r=1$ .}
\label{Fig:rate1}
\end{figure}

In order to study the effect of  various system parameters, namely $N_t, N_a,$ and $M$, on the overall system performance we summarize the BER and the throughput performances for both STBC-SM and MASM schemes  in Table \ref{tableII}. It can be noticed that the proposed scheme provides a significant enhancement on the BER performance compared to MASM for different values of the involved parameters. For instance, for the the case of $N_t=6, N_a=3,$ and $M=2$, the percentage difference in the BER is almost $99.9999 \%$ compared to MASM. Additionally, it is noted that the BER performance of MASM is degraded as $N_a$ increases which is attributed to the increased  spatial interference between the active LEDs.  Moreover, for fixed  $N_a$ and varying $N_t$, the BER performance for MASM drops noticeably because of the decreased separation and increased correlation between the LEDs. A similar remark can be also made for the case of increasing the modulation order $M$. On the contrary, STBC-SM is more robust to the change in the involved parameters in terms of the achievable BER.  An interesting remark can be also concluded from the last two settings, i.e., $N_t=7$ and $N_t=8$. Specifically, even though we obtain the same spectral efficiency from the two configurations for the same values of $N_a$ and $M$, it can be noticed that using lower $N_t$ would results in a lower BER values. This is attributed to the increased separation, and hence, reduced interference between LEDs. However, in terms of throughput, it can be seen that MASM outperforms STBC-SM particularly for low modulation size. For instance, the percentage difference between the two schemes for $N_t=6, N_a=4,$ and $M=2$ is $68.9655\%$. Therefore, it is concluded that the BER related  benefit of the proposed scheme is much greater compared to the throughput loss.
Also, even though MASM outperforms STBC-SM in the achieved throughput, it is clear from the setting $N_t=6, N_a=2,$ and $M=16$ that the  gap difference in the throughput decreases between both schemes for large modulation orders $M$. For instance, the percentage difference in terms of throughput for this case is $3.5088\%$ relative to the MASM, which is very small compared to the case of $N_t=6, N_a=2,$ and $M=4$, where the percentage difference is $45.3125\%$. Therefore, this renders the proposed scheme more suitable for setups where  the room has limited number of LEDs and high modulation order $M$. It is also noted here that using high-rate QOSTBC can assist towards  reducing the throughput gap if it is integrated with SM.
\begin{table*}[ht]
\centering
\caption{BER and throughput for different configurations } 
\begin{small}
\begin{tabular}{||c|c|c||c|c|c||c|c|c||}
\hline
\multicolumn{3}{||c||}{Involved Parameters} & \multicolumn{3}{|c||}{BER} &\multicolumn{3}{c||}{Throughput} \\
\hline
$N_t$ & $N_a$  & $M$ &STBC-SM & MASM &$\Delta \%$&$R_{stbc-sm}$& $R_{sm}$ & $\Delta \%$\\
\hline \hline $6$&$2$&$2$&$\approx 0$ & $0.024$ & $100$ & $2.5$ & $4.8 $ & $47.9167 $\\
\hline $6$&$3$&$2$&$ 2.4\times 10^{-9}$&$ 0.133$ & $99.9999 $ & $2.3 $& $ 6$ & $61.6667 $ \\
\hline $6$&$4$&$2$&$3.4 \times 10^{-6} $&$0.171 $ & $99.9980 $ & $ 1.8$ & $5.8 $ & $ 68.9655$\\
\hline \hline 
$6$&$2$&$2$&$ \approx  0$ &$0.024 $& $100 $ & $ 2.5$ & $4.8 $ & $47.9167 $\\
\hline $6$&$2$&$4$&$9.1 \times 10^{-4}$&$ 0.084$& $98.9167 $ & $ 3.5 $ & $ 6.4$ & $ 45.3125$ \\
\hline $6$&$2$&$16$&$0.005$ & $0.475$ &$98.9474 $  &  $ 5.5$ & $5.7 $ & $3.5088 $ \\
\hline \hline 
$4$&$2$&$2$&$\approx 0 $ &$8.9 \times 10^{-5} $ & $ 100$ & $ 2$ &  $4 $ & $ 50$\\
\hline $6$&$2$&$2$&$\approx  0 $&$ 0.024$ & $100 $ & $2.5 $ & $4.8 $ & $ 47.9167$\\
\hline $7$&$2$&$2$&$7.5\times 10^{-9}$&$0.065 $ & $99.9999 $ &$ 3$ & $ 5.6$& $46.4286 $ \\
\hline $8$&$2$&$2$&$ 1.3 \times 10^{-6} $&$0.079 $ & $99.9983 $ & $2.99 $& $ 5.5$ & $45.6363 $\\
\hline \hline
\end{tabular}
\end{small}
\label{tableII}
\end{table*}

It is well known that the performance of   indoor VLC networks is very sensitive to the placement of the LEDs in the room. Hence, Fig.\ref{Fig:seperation} demonstrates the effect of varying the horizontal and the vertical separations between LEDs 1, 2, 3, and 4. It can be seen from the figure that there exists an optimum separation between the LEDs where both schemes achieve the best performance. This is explained as follows: the effect of the correlation is very high when the LEDs are very close to each other, which increases  the error rate in the index detection. On the contrary, when the LEDs are placed considerably  far from each other,  the amount of the received power from each active LED is ultimately reduced. Therefore, an optimum placement of the LEDs which guarantees low correlation of the involved wireless channels and acceptable received amount of power is  necessary. It is also noted that the proposed system is more robust to the LEDs' placement and achieves a good BER performance for various separation values. 
\begin{figure}[ht]
\centering
\includegraphics[width=270pt, height = 7cm]{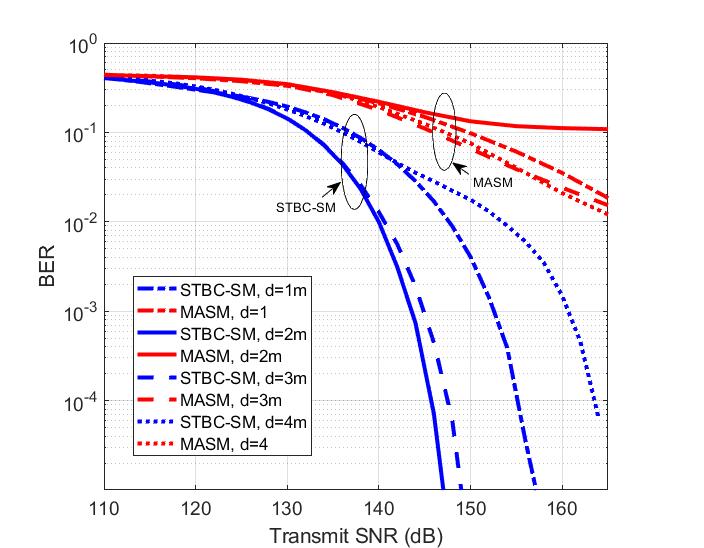}
\caption{Effect of varying the horizontal and the vertical separation $d$ between LEDs 1, 2, 3, and 4. $ N_t = 6$, $N_a = 2$, $N_r=1$, and $M=2$.}
\label{Fig:seperation}
\end{figure}\\ \indent
Beside this, we present the effect of varying the half-power semi-angle ($\varphi_{1/2}$) of the LEDs on the BER performance for different values of transmit SNR, while fixing the other parameters. Fig. \ref{Fig:varphi} shows the simulated and analytical BER versus the semi-angle for $N_t=6$, $N_a=2$, $M=2$, and SNR values of 130 dB, 135 dB, and 140 dB. As expected, the MIMO system performance in VLC is highly dependent upon the value of $\varphi_{1/2}$ of the LEDs. As depicted in the figure, by fixing all other system parameters, the performance ultimately improves as the angle decreases due to the improved channel gains and the reduced correlation between different channels. 
\begin{figure}[ht]
\centering
\includegraphics[width=270pt, height = 7cm]{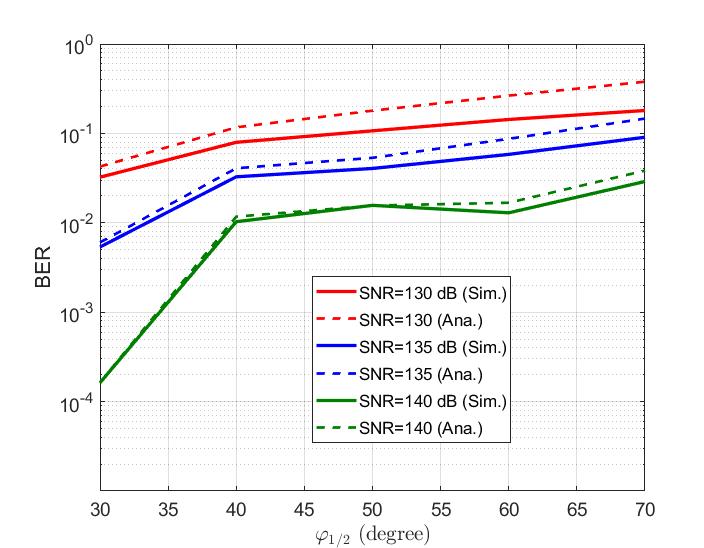}
\caption{Effect of varying $\varphi_{1/2}$ for different transmit SNR. $N_t = 6$, $N_a = 2$, $N_r=1$, and $M=2$.}
\label{Fig:varphi}
\end{figure}

 Finally, we study the BER performance of the proposed high-rate QOSTBC-SM and compare its  achievable performance with respect to MASM and RC-SM based schemes for different spectral efficiency values, i.e.,  $\eta$ = 5 bpcu, 6 bpcu, and 7 bpcu. It can be seen from Fig. \ref{Fig:QOSTBC_vs_MASM} that integrating QOSTBCs with SM enhances the system performance in terms of the achievable BER compared to MASM and RC-SM for fixed spectral efficiency values. This performance improvement can be noticed more profoundly at higher spectral efficiency values. For instance, a $\eta$ = 7 bpcu spectral efficiency value was realized through QOSTBC-SM by choosing the system setup as follows: $N_t=4, N_a=2,$ and $M=8$, whereas the achieved error rate at SNR $=170$ dB is $1.6\times 10^{-3}$. 
 On the contrary, in order to achieve the same spectral efficiency using RC-SM, the system setup was chosen as, $N_t=4, N_a=2,$ and $M=32$. However,  the achieved error rate in this case considering the same SNR is $1.4\times10^{-2}$. The need for higher $M$-ary modulation values while fixing other system setup parameters explains the error rate performance gap between the two schemes at higher $\eta$ values. However, at lower spectral efficiency values, i.e., 5 bpcu, RC-SM achieves better performance compared to both QOSTBC-SM and MASM. Besides this, since the proposed QOSTBC matrix has a rate of $R=2$, it provides a spectral efficiency enhancement when integrated with SM compared to OSTBC-SM. Yet, this spectral efficiency enhancement comes at the expense of a slightly increased receiver complexity, which as already mentioned it is tolerable. 
\begin{figure}[ht]
\centering
\includegraphics[width=270pt, height = 7cm]{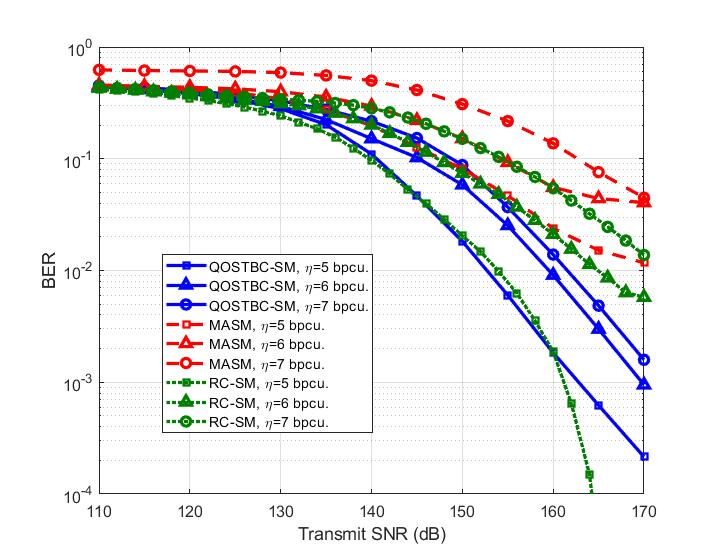}
\caption{BER performance comparison between the proposed QOSTBC-SM, MASM, and RC-SM. $\eta$ = 5 bpcu, 6 bpcu, and 7 bpcu.}
\label{Fig:QOSTBC_vs_MASM}
\end{figure}

\section {Conclusion} 
In this paper, we considered the performance of a high-rate low complexity MIMO transmission  STBC-SM scheme for VLC systems. A general technique has been proposed  based on constructing any STBC-SM scheme for any $M$-ary PAM size and any number of transmitting LEDs. It has been shown through out simulations  and derived upper bounds that STBC-SM is a promising MIMO technique for indoor VLC system as it offers  improved system BER performance, throughput and reduced receiver complexity compared to MASM and RC-SM. To achieve that, part of the bit stream was used to activate a group of the available LEDs, while the other part of the bit stream is conveyed through space time coded intensity modulation. Furthermore, it has been shown through extensive simulations that STBC-SM maintains a rather good performance for various  setups, whilst an improved throughput is achieved compared to MASM as the modulation size increases. 

\bibliographystyle{IEEEtran}
\balance 
\bibliography{ref}

\begin{IEEEbiography}
[{\includegraphics[width=1in,height=1.25in,clip,keepaspectratio]{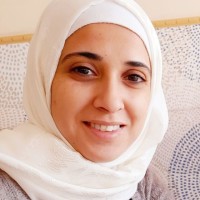}}]
{Shimaa Naser} received the M.Sc. degree in Electrical Engineering from the University of Science and Technology, Irbid, Jordan in 2015. Since 2018, she is pursuing her Ph.D. degree at the Department of Electrical and Computer Engineering at Khalifa University, Abu Dhabi, UAE. Her research interests  include  advanced  digital  signal  processing and modulation techniques for visible light communications, MIMO-based communiaction, and orthogonal/non-orthogonal  multiple  access.
\end{IEEEbiography}

\begin{IEEEbiography}
[{\includegraphics[width=1in,height=1.25in,clip,keepaspectratio]{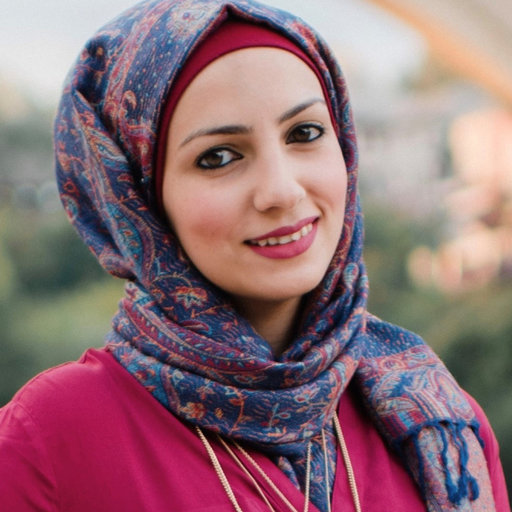}}]
{Lina Bariah} (S’13-M’19) received the M.Sc. and Ph.D degrees in communications engineering from Khalifa University, Abu Dhabi, United Arab Emirates, in 2015 and 2018. She is currently a Postdoctoral fellow with the KU Center for Cyber-Physical Systems, Khalifa University, UAE. She was a Visiting  Researcher with the Department of Systems and Computer Engineering, Carleton University, ON, Canada, in summer 2019. Her research interests include advanced digital signal processing techniques for communications, channel estimation, cooperative communications, non-orthogonal multiple access, cognitive radios, and intelligent surfaces.

\end{IEEEbiography}

\begin{IEEEbiography}
[{\includegraphics[width=1in,height=1.25in,clip,keepaspectratio]{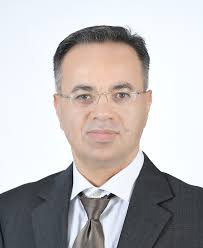}}]
{Sami Muhaidat} (S’01-M’08-SM’11) received the Ph.D. degree in electrical and computer engineering from the University of Waterloo, Waterloo, ON, Canada, in 2006. From 2007 to 2008, he was an NSERC Post-Doctoral Fellow with the Department of Electrical and Computer Engineering, University of Toronto, ON, Canada. From 2008 to 2012, he was Assistant Professor with the School of Engineering Science, Simon Fraser University, Burnaby, BC, Canada. He is
currently an Associate Professor with Khalifa University, Abu Dhabi, UAE, and a Visiting Professor with the Department of Electrical and Computer Engineering, University of Western
Ontario, London, ON, Canada. He is also a Visiting Reader with the Faculty of Engineering, University of Surrey, Guildford, U.K. Dr. Muhaidat currently serves as an Area Editor of the IEEE TRANSACTIONS ON COMMUNICATIONS, and he was previously a Senior Editor of the IEEE COMMUNICATIONS LETTERS and an Associate Editor of IEEE TRANSACTIONS ON COMMUNICATIONS, IEEE COMMUNICATIONS LETTERS and IEEE TRANSACTIONS ON VEHICULAR TECHNOLOGY. He was a recipient of several scholarships during his undergraduate and graduate studies and the winner of the 2006 NSERC PostDoctoral Fellowship Competition. Dr Muhaidat is a Senior Member IEEE.
\end{IEEEbiography}

\begin{IEEEbiography}
[{\includegraphics[width=1in,height=1.25in,clip,keepaspectratio]{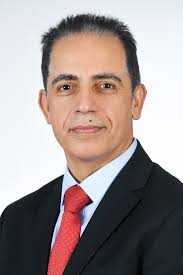}}]
{Mahmoud Al-Qutayri} (S'87, M'92, SM'02) is the Associate Dean for Graduate Studies – College of Engineering, and a Professor of Electrical and Computer Engineering at Khalifa University, UAE. He received the B.Eng., MSc and PhD degrees from Concordia University, Canada, University of Manchester, U.K., and the University of Bath, U.K., all in Electrical and Electronic Engineering in 1984, 1987, and 1992, respectively. Prior to joining Khalifa University, he was a Senior Lecturer at De Montfort University, UK. This was preceded by a Research Officer appointment at University of Bath, UK. He has published numerous technical papers in peer reviewed international journals and conferences. He coauthored a book as well as a number of book chapters. His main fields of research include embedded systems design, applications and security, design and test of mixed-signal integrated circuits, wireless sensor networks, and cognitive wireless networks. During his academic career, Dr. Al-Qutayri made many significant contributions to both undergraduate as well as graduate education. His professional service includes membership of the steering, organizing and technical program committees of many international conferences.
\end{IEEEbiography}

\begin{IEEEbiography}
[{\includegraphics[width=1in,height=1.25in,clip,keepaspectratio]{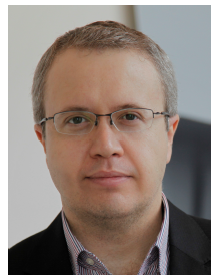}}]
{Murat Uysal} (Fellow, IEEE) received the B.Sc. and the M.Sc. degree in electronics and communication engineering from Istanbul Technical University, Istanbul, Turkey, in 1995 and 1998, respectively, and the Ph.D. degree in electrical engineering from Texas A\&M University, College Station, Texas, in 2001. He is currently a Full Professor and Chair of the Department of Electrical and Electronics Engineering at Özyeğin University, Istanbul, Turkey. He also serves as the Founding Director of Center of Excellence in Optical Wireless Communication Technologies (OKATEM). Prior to joining Özyeğin University, he was a tenured Associate Professor at the University of Waterloo, Canada. Dr. Uysal's research interests are in the broad area of communication theory with a particular emphasis on the physical layer aspects of wireless communication systems in radio and optical frequency bands. On these topics, he has authored some 350 journal and conference papers and received more than 14.000 citations with an h-index of 56.
Prof. Uysal is an IEEE Fellow and the former Chair of IEEE Turkey Section. He currently serves on the editorial board of IEEE Transactions on Wireless Communications. In the past, he served as an Editor for IEEE Transactions on Communications, IEEE Transactions on Vehicular Technology, and IEEE Communications Letters. He was involved in the organization of several IEEE conferences at various levels. In particular, he served as the Technical Program Committee Chair of major IEEE conferences including WCNC 2014, PIMRC 2019 and VTC-Fall 2019.
Prof Uysal’s major distinctions include NSERC Discovery Accelerator Award in 2008, University of Waterloo Engineering Research Excellence Award in 2010, Turkish Academy of Sciences Distinguished Young Scientist Award in 2011, Özyeğin University Best Researcher Award in 2014, National Instruments Engineering Impact Award in 2017, Elginkan Foundation Technology Award in 2018 and IEEE Communications Society Best Survey Paper Award in 2019 among others.

\end{IEEEbiography}

\begin{IEEEbiography}
[{\includegraphics[width=1in,height=1.25in,clip,keepaspectratio]{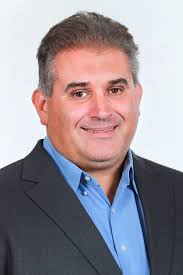}}]
{Paschalis C. Sofotasios} (S’07-M’12-SM’16) was born in Volos, Greece, in 1978. He received the M.Eng. degree from Newcastle University, U.K., in 2004, the M.Sc. degree from the University of Surrey, U.K., in 2006, and the Ph.D. degree from the University of Leeds, U.K., in 2011. He has held academic positions at the
University of Leeds, U.K., University of California at Los Angleles, CA, USA, Tampere University of Technology, Finland, Aristotle University of
Thessaloniki, Greece and Khalifa University of Science and Technology, UAE, where he currently serves as Assistant Professor in the department of Electrical Engineering and Computer Science.
His M.Sc. studies were funded by a scholarship from UK-EPSRC and his Doctoral studies were sponsored by UK-EPSRC and Pace plc. His research interests are in the broad areas of digital and optical wireless communications as well as in topics relating to special functions and
statistics. Dr. Sofotasios serves as a regular reviewer for several international journals and has been a member of the technical program
committee of numerous IEEE conferences. He currently serves as an Editor for the IEEE COMMUNICATIONS LETTERS and he received the Exemplary Reviewer Award from the IEEE COMMUNICATIONS LETTERS
in 2012 and the IEEE TRANSACTIONS ON COMMUNICATIONS in 2015 and 2016. Dr. Sofotasios is a Senior Member IEEE and he received the Best Paper Award at ICUFN 2013. 
\end{IEEEbiography}

\end{document}